\def\year{2018}\relax
\documentclass[letterpaper]{article} %DO NOT CHANGE THIS
\usepackage{aaai18}  %Required
\usepackage{times}  %Required
\usepackage{helvet}  %Required
\usepackage{courier}  %Required
\usepackage{url}  %Required
\usepackage{graphicx}  %Required
\usepackage{pgfplots}
\usepackage{tikzscale}
\usepgfplotslibrary{groupplots}
\usepackage{color}
\usepackage{amsmath,amssymb,amsthm}
\theoremstyle{plain}
\newtheorem{theorem}{Theorem}
\newtheorem{lemma}[theorem]{Lemma}
\newtheorem{proposition}[theorem]{Proposition}
\newtheorem{corollary}[theorem]{Corollary}
\theoremstyle{definition}
\newtheorem{definition}{Definition}
\newtheorem{example}{Example}
\newtheorem{remark}{Remark}
\usepackage[ruled,vlined,linesnumbered]{algorithm2e}
\SetKw{Continue}{continue}
\usepackage{caption} 
\usepackage{subcaption}
\usepackage{xspace}
\usepackage{adjustbox}
\usepackage{textcomp}
\usepackage{mathtools}
\usepackage{multirow}
\usepackage{IEEEtrantools}
\usepackage{tikz}
\usetikzlibrary{automata,arrows,shapes}
\usepackage{mathrsfs}
\pgfdeclarelayer{background}
\pgfsetlayers{main,background}
\tikzstyle{vertex}=[circle,fill=black!10,minimum size=20pt,inner sep=0pt]
\tikzstyle{edge} = [draw,thick]
\tikzstyle{weight} = [font=\small]
\tikzstyle{redundant edge} = [draw,line width=5pt,-,red,opacity=.7]

\usepackage{pgfplots}
\pgfplotsset{width=.3\linewidth,compat=1.9}

\providecommand\given{}
\newcommand\SetSymbol[1][]{\nonscript\:#1\vert\nonscript\:
  \mathopen{}\allowbreak}
\DeclarePairedDelimiterX\Set[1]\{\}{%
  \renewcommand\given{\SetSymbol[\delimsize]}
  #1 }

\newcommand{\la}{\langle}
\newcommand{\ra}{\rangle}
\newcommand{\comp}{\otimes}
\newcommand{\AC}{\ensuremath{\mathsf{ACSTP}}\xspace}
\newcommand{\DAC}{\ensuremath{\mathsf{DisACSTP}}\xspace}

\newcommand{\PPC}{\ensuremath{\mathsf{P\textsuperscript{3}C}}\xspace}
\newcommand{\DPPC}{\ensuremath{\mathsf{D\triangle PPC}}\xspace}
\newcommand{\N}{\mathcal{N}}
\newcommand{\M}{\mathcal{M}}
\newcommand{\network}{\mathcal{N}= \la V,D,C \ra}

\newcommand{\graph}{G= (V,E)}

\newcommand{\magenta}[1]{{#1}}
\newcommand{\blue}[1]{\textcolor{black}{#1}}
\newcommand{\red}[1]{\textcolor{black}{#1}}

\usepackage[hidelinks,draft=true]{hyperref}
\AtBeginDocument{}
\setkeys{Gin}{draft=false}      % Show figures in draft mode
\usepackage[obeyDraft,colorinlistoftodos]{todonotes}

\colorlet{JL}{green!20!yellow!60!white} \colorlet{JLline}{JL!80!black}

\colorlet{SL}{red!20!yellow!60!white}
\colorlet{SLline}{SL!80!black}

\colorlet{SK}{blue!10!white}
\colorlet{SKline}{SK!80!black}

\begin{document}
% The file aaai.sty is the style file for AAAI Press proceedings,
% working notes, and technical reports.
%
\title{Multiagent Simple Temporal Problem: The Arc-Consistency
  Approach}
\author{Shufeng Kong$^1$ \and  Jae Hee Lee$^1$ \and Sanjiang Li$^{1,2}$\\
  $^1$Centre for Quantum Software and Information, FEIT, University of Technology Sydney, Australia\\
  $^2$UTS-AMSS Joint Research Laboratory, AMSS, Chinese Academy of Sciences, China\\
  shufeng.kong@student.uts.edu.au, \{jaehee.lee,
  sanjiang.li\}@uts.edu.au%
}

\maketitle
\begin{abstract}

  The Simple Temporal Problem (STP) is a fundamental temporal
  reasoning problem and has recently been extended to the Multiagent
  Simple Temporal Problem (MaSTP). In this paper we present a novel
  approach that is based on enforcing arc-consistency (AC) on the
  input (multiagent) simple temporal network. We show that the
  AC-based approach is sufficient for solving both the STP and MaSTP
  and provide efficient algorithms for them. As our AC-based approach
  does not impose new constraints between agents, it does not violate
  the privacy of the agents and is superior to the state-of-the-art
  approach to MaSTP. Empirical evaluations on diverse benchmark
  datasets also show that our AC-based algorithms for STP and MaSTP
  are significantly more efficient than existing approaches.
  
\end{abstract}

\section{Introduction}

The Simple Temporal Problem (STP)~\cite{DechterMP91} is arguably the
most well-known quantitative temporal representation framework in AI.
The STP considers time points as the variables and represents temporal
information by a set of unary or binary constraints, each specifying
an interval on the real line. Since its introduction in 1991, the STP
has become an essential sub-problem in planning or scheduling
problem~\cite{bartak_introduction_2014}.

While the STP is initially introduced for a single scheduling agent
and is solved by centralized algorithms, many real-world applications
involve multiple agents who interact with each other to find a
solution like the following example:

\begin{example}\label{example}
  When Alice is looking for a position at company $X$, she might need
  to arrange an interview appointment with $X$. Suppose that her
  colleague Bob is also applying for the position and Alice and Bob
  are both applying for another position at company $Y$. To
  represent and solve such \blue{an} interview scheduling problem, we need a
  multiagent framework \magenta{(see Figure~\ref{fig:mastp} for an
    illustration)}.
\end{example}
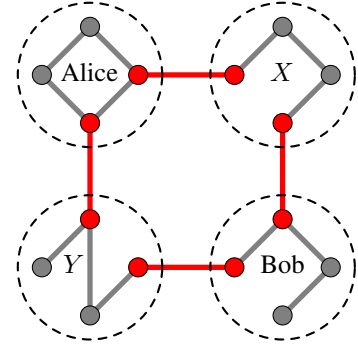
\begin{figure}[t]
  \centering
  \definecolor{yqyqyq}{rgb}{0.5019607843137255,0.5019607843137255,0.5019607843137255}
  \definecolor{ffqqqq}{rgb}{1.,0.,0.}
  \begin{tikzpicture}[line cap=round,line join=round,>=triangle
    45,x=.8cm,y=.8cm, scale=0.8]
    \clip(1.5, 1.45) rectangle (10,9);%
    \draw [line width=2.pt,color=yqyqyq] (3.,7.)-- (4.,8.); \draw
    [line width=2.pt,color=yqyqyq] (4.,8.)-- (5.,7.); \draw [line
    width=2.pt,color=yqyqyq] (5.,7.)-- (4.,6.); \draw [line
    width=2.pt,color=yqyqyq] (4.,6.)-- (3.,7.); \draw [line
    width=2.pt,,color=ffqqqq] (5.,7.)-- (7.,7.); \draw [line
    width=2.pt,,color=ffqqqq] (4.,6.)-- (4.,4.); \draw [line
    width=2.pt,,color=ffqqqq] (8.,6.)-- (8.,4.); \draw [line
    width=2.pt,,color=ffqqqq] (5.,3.)-- (7.,3.); \draw [line
    width=2.pt,color=yqyqyq] (8.,6.)-- (9.,7.); \draw [line
    width=2.pt,color=yqyqyq] (9.,7.)-- (8.,8.); \draw [line
    width=2.pt,color=yqyqyq] (8.,8.)-- (7.,7.); \draw [line
    width=2.pt,color=yqyqyq] (8.,2.)-- (9.,3.); \draw [line
    width=2.pt,color=yqyqyq] (9.,3.)-- (8.,4.); \draw [line
    width=2.pt,color=yqyqyq] (8.,4.)-- (7.,3.); \draw [line
    width=2.pt,color=yqyqyq] (3.,3.)-- (4.,4.); \draw [line
    width=2.pt,color=yqyqyq] (4.,4.)-- (4.,2.); \draw [line
    width=2.pt,color=yqyqyq] (4.,2.)-- (5.,3.); \draw (7.55,7.45)
    node[anchor=north west] {$X$}; \draw (3.2,7.45) node[anchor=north
    west] {$\text{Alice}$}; \draw (7.35,3.45) node[anchor=north west]
    {$\text{Bob}$}; \draw (3.25,3.45) node[anchor=north west] {$Y$};
    \draw [line width=0.8pt,dash pattern=on 3pt off 3pt] (4.,7.)
    circle (1.5); \draw [line width=0.8pt,dash pattern=on 3pt off 3pt]
    (8.,7.) circle (1.5); \draw [line width=0.8pt,dash pattern=on 3pt
    off 3pt] (4.,3.) circle (1.5); \draw [line width=0.8pt,dash
    pattern=on 3pt off 3pt] (8.,3.) circle (1.5);
    \begin{scriptsize}
      \draw [fill=ffqqqq] (5.,3.) circle (4.5pt); \draw [fill=ffqqqq]
      (4.,4.) circle (4.5pt); \draw [fill=ffqqqq] (7.,3.) circle
      (4.5pt); \draw [fill=yqyqyq] (3.,3.) circle (4.5pt); \draw
      [fill=yqyqyq] (4.,2.) circle (4.5pt); \draw [fill=yqyqyq]
      (8.,2.) circle (4.5pt); \draw [fill=ffqqqq] (8.,4.) circle
      (4.5pt); \draw [fill=yqyqyq] (9.,3.) circle (4.5pt); \draw
      [fill=ffqqqq] (4.,6.) circle (4.5pt); \draw [fill=ffqqqq]
      (5.,7.) circle (4.5pt); \draw [fill=ffqqqq] (7.,7.) circle
      (4.5pt); \draw [fill=ffqqqq] (8.,6.) circle (4.5pt); \draw
      [fill=yqyqyq] (9.,7.) circle (4.5pt); \draw [fill=yqyqyq]
      (8.,8.) circle (4.5pt); \draw [fill=yqyqyq] (4.,8.) circle
      (4.5pt); \draw [fill=yqyqyq] (3.,7.) circle (4.5pt);
    \end{scriptsize}
  \end{tikzpicture}
  \caption{An illustration of Example~\ref{example}. Alice, Bob, company
    $X$, company $Y$ are four agents, each owning a local simple
    temporal network. The circles represent variables and edges
    constraints. Red edges represent constraints that are shared by
    two different agents.
    % The shared variables are marked as light green
    % nodes, and the external constraints are marked as dashed red
    % arrows. The red arrows form a circle. To triangulate this cycle,
    % we need add at least one edge from Alice to Bob or one edge from X
    % to Y.
  }\label{fig:mastp}
\end{figure}

Recently, the extension of STP to multiagent STP (MaSTP) has been
provided in \cite{DBLP:journals/jair/BoerkoelD13}, which presents a
formal definition of the MaSTP as well as a distributed algorithm,
called \DPPC, for computing the complete joint solution space.

However, as \DPPC is based on the \PPC algorithm~\cite{PlankenWK08},
which triangulates the input constraint graph, it has the drawback of
creating new constraints between agents \magenta{that are possibly not
  directly connected}. In Figure~\ref{fig:mastp}, \DPPC triangulates
the inner cycle by adding at least one new constraint either between
$X$ and $Y$ or between Alice and Bob. Neither of these new constraints
are desirable, as they introduce constraints between two previously
not directly connected agents and thus present a threat to the privacy
of the relevant agents.

As the recent technological advancements have allowed for solving
larger problems that are highly interwoven and dependent on each
other, efficiency and privacy have become critical requirements. To
address this challenge, we propose a new approach to solve the MaSTP,
which is based on \emph{arc-consistency}.

A constraint $R$ between two variables $x,y$ is called arc-consistent
(AC), if for every value $d_x$ from the domain of $x$ there is a
\magenta{value} $d_y$ in the domain of $y$ such that $(d_x,d_y)\in R$.
While AC is an important tool for solving finite (multiagent)
constraint satisfaction problems (CSPs)
\red{\cite{montanari1974networks,baudot1997analysis,nguyen1998distributed,hamadi2002optimal}}
at first glance it is not clear how it can be applied to solving CSPs
with real domains such as the STP, \red{because either the existing
  AC algorithms are fine-grained and work with each single element of
  a domain to enforce AC, which is impossible for real domains, or
  they are coarse-grained, but cannot guarantee their termination, as
  real domains can be infinitely refined when constraints are
  propagated.}\footnote{\red{\citeauthor{DechterMP91}
    \shortcite{DechterMP91} for example, suggest discretizing the
    domains to overcome this issue, in which case the total number of
    contraint propagations would depend on the sizes of the domains.
    The performance of our AC algorithm for (multiagent) STP does not
    depend on the sizes of the domains.}}

Our contributions in this paper are as follows.

\begin{itemize}
\item We provide the first AC-based approach for solving STP and
  analyze its computational complexity
  (Section~\ref{sec:solving-stp-with}).
\item We provide the first AC-based approach for solving multiagent
  STP, which preserves the privacy of the agents, and analyze its
  computational complexity (Section~\ref{sec:solving-mastp-with}).
\item We experimentally show that both our centralized and distributed
  algorithms outperform \blue{their} existing counterparts for solving STP
  (Section~\ref{sec:evaluation}).
\end{itemize}

The next section gives a formal introduction to STP.

% and is superior to \DPPC, the state-of-the art approach ot MaSTP.

% The remainder of this paper is organized as follows. In Section 2,
% we provide basic notions and properties related to the STP and the
% MaSTP and then present and analyze the centralized and distributed
% AC algorithms for the STP and, respectively, the MaSTP in Sections 3
% and 4. Section 5 then evaluates our approach by using large
% benchmark datasets. Finally, Section 6 concludes the paper
% \red{and points out several potential ways to optimize the
% algorithm.}

% While this is indeed true for qualitative CSPs, we will shortly see
% that an AC algorithm is sufficient for deciding if an STP is
% satisfiable and for constructing such a solution if it is
% satisfiable.

% For the MaSTP, this is especially important, as we do not need to
% add new edges and thus the structure of the underlying constraint
% graph is unchanged and the privacy and autonomy of individual agents
% is protected as much as possible.

% \section{Preliminary}
\section{The Simple Temporal Problem}
\label{sec:preliminaries}

% This section briefly introduces the STP and the Multiagent STP
% (MaSTP); details can be found in \cite{DechterMP91} and
% \cite{DBLP:journals/jair/BoerkoelD13}, respectively.
This section briefly introduces the STP. Details can be found in
\cite{DechterMP91}.

% \subsection{The Simple Temporal Problem}
The \emph{simple temporal problem} (STP) is a constraint satisfaction
problem where each constraint is a set of linear inequalities of the
form
\begin{equation}\label{eq:1.1}
  a_{vw} \le w - v \le b_{vw},
\end{equation}
where $a_{vw},b_{vw}$ are constants and $v, w$ are variables defined
on a continuous domain representing time points. The constraint in
\eqref{eq:1.1} is abbreviated as $I_{vw}= [a_{vw}, b_{vw}]$. As
\eqref{eq:1.1} is equivalent to $-b_{vw} \le v - w \le -a_{vw}$, we
also obtain $I_{wv} = I^{-1}_{vw} = [-b_{vw}, -a_{vw}]$. The
\emph{domain} of each variable $v$ is an interval $I_v=[a_v,b_v]$,
where $I_v$ could be a singleton or empty. Assume that $o$ is a
special \emph{auxiliary} variable that represents the fixed zero
temporal point. Then the domain $I_v$ can also \magenta{be} regarded
as a constraint from $o$ to $v$ and
$I_v = [a_v, b_v] = [a_{ov}, b_{ov}] = I_{ov}$.
% ======= The \emph{simple temporal problem} (STP) is a constraint
% satisfaction problem where each constraint is a set of linear
% inequalities of the form
% \begin{equation}\label{eq:1.1}
%     a_{vw} \le w - v \le b_{vw},
%   \end{equation}
%   where $a_{vw},b_{vw}$ are constants and $v, w$ are variables
%   defined on a continuous domain representing time points. The
%   constraint in \eqref{eq:1.1} is abbreviated as
%   $I_{vw}= [a_{vw}, b_{vw}]$. As \eqref{eq:1.1} is equivalent to
%   $-b_{vw} \le v - w \le -a_{vw}$, we also obtain
%   $I_{ji} = I^{-1}_{ij} = [-b_{vw}, -a_{vw}]$. The \emph{domain} of
%   each variable $v$ is an interval $I_v=[a_v,b_v]$, where $I_v$
%   could be a singleton or empty. Assume that $o$ is a special
%   \emph{auxiliary} variable that represents the fixed zero temporal
%   point. Then the domain $I_v$ can also be regarded as a constraint
%   from $o$ to $v$ and
%   $I_v = [a_v, b_v] = [a_{ov}, b_{ov}] = I_{ov}$. >>>>>>>
%   834b0305ce1008eeca26b571cd760f605f3bc313
  
Algebraic operations on STP constraints are defined as follows. The
\emph{intersection} of two STP constraints defined on variables $v,w$
yields a new constraint over $v,w$ that represents the conjunction of
the constraints. It is defined as
$I_{vw} \cap I'_{vw} := [\max\{a_{vw}, a'_{vw}\}, \min\{b_{vw},
b'_{vw}\}]$.

The \emph{composition} of an STP constraint $I_{vu}$ over variables
$v,u$ and another STP constraint $I_{uw}$ over $u,w$ yields a new STP
constraint over $v,w$ that is inferred from the other two constraints
and is defined as
$I_{vu} \comp I_{uw}:= [a_{vu} + a_{uw}, b_{vu} + b_{uw}]$. Here we
require that $[a,b]\comp \varnothing=\varnothing$ for any $a\leq b$.

\begin{remark}
  For STP constraints, the composition and intersection are
  associative and, as noted in~\cite{DechterMP91}, composition
  distributes over \emph{non-empty} intersection for intervals, i.e.,
  $I \comp (J \cap K) = (I \comp J) \cap (I \comp K)$ for any three
  intervals $I,J,K$ such that $J\cap K\not=\varnothing$.

  % Composition does \emph{not} always distributes over \emph{empty}
  % intersection for intervals:
  % $([0,10] \comp [1,2])\cap ([0,10] \comp [3,4])= [1,12] \cap [3,14]
  % \not=\varnothing$, but
  % $([0,10] \comp ([1,2]\cap [3,4]) = [0,10] \comp \varnothing
  % =\varnothing$.
  
\end{remark}

\begin{definition}
  An instance of STP is called a \emph{simple temporal network} (STN)
  and is a tuple $\la V,D,C\ra$, where $V$ is a finite set of
  variables, $D=\{I_v \mid v \in V\}$ is a set of intervals, and $C$
  is a set of STP constraints defined on $V$.
\end{definition}
\noindent We assume that all variables in $V$ appear in $C$ and at
\emph{most} one constraint exists between any pair of variables $v$ and
$w$. Moreover, if $I_{vw}=[a,b]$ is the constraint in $C$ from $v$ to
$w$, we always assume that the constraint $I_{\blue{wv}}=I_{vw}^{-1}=[-b,-a]$
is also in $C$. As previously mentioned, the domain $I_v$ of each
variable $v$ can be regarded as either a unary constraint, or a binary
constraint $I_{ov}=I_v$, where $o$ is a fixed variable representing
the zero time point.
  
An STN naturally induces a graph in the following sense.

 \begin{definition}
   The \emph{constraint graph} $G_{\mathcal{N}} = (V,E)$ of an STN
   $\network$ is an undirected graph, where the set $E$ of edges
   consists of constrained unordered pairs of variables in $C$, i.e.,
   \[
     E = \Set*{\{v,w\} \given v,w \in V, v \ne w, I_{vw} \in C}.
   \]
   % For simplicity, we sometimes write $e_{ij}$ for an edge in place
   % of
   % $\{v,w\}$ and we say $e_{ij} \in E$ if either $e_{ij}\in E$ or
   % $e_{ji}\in E$.
 \end{definition}

 Let $G_{\mathcal{N}} = (V,E)$ be the constraint graph of an STN $\N$.
 We can use a labelled \emph{directed} graph to illustrate \red{$\N$}, % $G_{\mathcal{N}}$
 where for any undirected edge $\{v,w\} \in E$ there is exactly one
 directed edge $(v, w)$ that is labelled with the corresponding
 interval $[a_{vw}, b_{vw}]$.

 A \emph{path} $\pi$ from $v$ to $w$ in $G_{\mathcal{N}}$ is a
 sequence of variables $u_{0},u_{1},...,u_{k}$ such that $v = u_{0}$,
 $w = u_{k}$, and $\{u_{s}, u_{s+1}\}$ is an edge in $E$ for each
 $s = 0, \dotsc, k-1$ \red{($k$ is called the length of $\pi$)}. We
 write $\bigotimes \pi$ for the composition of all these
 $I_{u_s,u_{s+1}}$, i.e.,
 \begin{align}
   \label{eq:cpi}
   \bigotimes \pi &= I_{u_0,u_{1}} \otimes I_{u_1,u_{2}} \otimes ... \otimes I_{u_{k-1},u_{k}}
   % &= [\Sigma_{0\leq s < k} l_{s,s+1}, \Sigma_{0\leq s<k} u_{s,s+1}]
   % .
 \end{align}
 If $v = w$, then we call $\pi$ a \emph{cycle} at $v$. \magenta{For a
   cycle $\pi$, let $[a, b] = \bigotimes\pi$.} We call $\pi$ a
 \emph{negative cycle} if  $b < 0$.

\begin{definition}
  A \emph{solution} of an STN $\network$ is an assignment, that
  assigns to each variable $v \in V$ a time point from $I_v \in D$
  such that all constraints in $C$ are satisfied. $\N$ is said to be
  \emph{consistent} if $\N$ has a solution. Two STNs are
  said to be \emph{equivalent} if they have the same solution set.
\end{definition}

\begin{definition}[Minimality]
  Let $\network$ be a consistent STN and let $v$ and $w$ be variables
  in $V$. A constraint $I_{vw}$ from $v$ to $w$ is said to be
  \emph{minimal} if every assignment that assigns time points from
  domains $I_v$ and $I_w$ to $v$ and $w$, respectively, and satisfies
  $I_{vw}$ can be extended to a solution of $\N$. A domain $I_v$ of
  $v \in V$ is said to be minimal if every assignment of a time point
  from $I_v$ to $v$ can be extended to a solution of $\N$. We say $\N$
  is minimal if every constraint in $C$ as well as every domain in $D$
  is minimal (\magenta{note that, since we regard domains as
    constraints between the zero time point $o$ and variables, we also
    require the domains to be minimal}).
\end{definition}

\section{Solving the STP with
  Arc-Consistency}\label{sec:solving-stp-with}

In this section we show that enforcing arc-consistency is sufficient
to solve the STP.

\begin{definition}
  Let $\network$ be an STN. Suppose $v$ and $w$ are two variables in
  $V$, $I_v$ and $I_w$ are, respectively, their domains, and $I_{vw}$
  is a constraint in $C$ from $v$ to $w$. We say that $I_{vw}$ is
  \emph{arc-consistent} (AC) (relative to $I_v$ and $I_w$) if for any
  $t_v \in I_v$ there exists some $t_w \in I_w$ such that
  $ t_w-t_v \in I_{vw}$, i.e., $a_{vw}\le t_w-t_v \le b_{vw}$. We say
  that $\N$ is AC if both $I_{vw}$ and $I_{wv}$ are AC for every
  constraint $I_{vw} \in C$.

  An STN $\mathcal{N}'=\la V, D', C\ra$ with $D'=\{I'_v \mid v\in V\}$
  is called the AC-\emph{closure} of $\N$, if $\mathcal{N}'$ is the
  largest arc-consistent STN which is equivalent to $\N$, in the sense that for
  every other arc-consistent STN $\mathcal{N}''=\la V, D'', C\ra$ with
  $D''=\{I''_i \mid v\in V\}$, we have that $I''_v \subseteq I'_v$ for all $v\in V$.
\end{definition}

\begin{lemma}\label{lemma:AC-0}
  Let $\network$ be an STN and $v,w \in V$ two variables that are
  constrained by $I_{vw}$ in $C$. Then $I_{vw}$ is arc-consistent
  relative to $I_v$ and $I_w$ iff
  $I_v\subseteq \red{I_w\comp I_{wv}}$.
\end{lemma}
% \JL{I think we need to require in this lemma that the domains are not
%   empty. Or we should in general avoid the trivial case of AC, where
%   the a domain or a relation is empty.}
\begin{proof}
  \red{%
    It suffices to show that
    \begin{equation}
      I_w\comp I_{wv} = \Set*{x \in \mathbb{R} \given \exists y \in
        I_w \text{ s.t.\ }  y - x
        \in I_{vw}}\label{eq:1}
    \end{equation}
    Let $I_v=[a,b]$, $I_w=[c,d]$ \red{and $I_{vw} = [e,f]$}. Then
    \begin{IEEEeqnarray*}{rCl}
      \IEEEeqnarraymulticol{3}{l}{\Set*{x \in \mathbb{R} \given \exists y \in I_w \text{ s.t.\ } y - x \in I_{vw}}}\\
      &=& \Set*{x \in \mathbb{R} \given \exists c \le y \le d \text{ s.t.\ }  e \le y - x \le f}\\
      &=& \Set*{x \in \mathbb{R} \given \exists c \le y \le d \text{ s.t.\ } y - f \le x \le y - e}\\
      &=& \Set*{x \in \mathbb{R} \given c - f \le x \le  d - e}\\
      &=& [c, f] \comp [-f, -e] = I_w\comp I_{wv},
    \end{IEEEeqnarray*}%
    which proves Eq.~\eqref{eq:1}.
  } 
  % Since $I_{vw}$ is AC relative to $I_v$ and $I_w$, for any
  % $x\in [a,b]$, there exists $y\in [c,d]$ such that $y-x\in [e,f]$,
  % i.e., $y-f\leq x\leq y-e$. Thus for any $x\in [a,b]$ we have
  % $c-f\leq x\leq d-e$, i.e., $x\in [c,d] \comp [-f,-e]$. On the other
  % hand, if $[a,b]\subseteq [c,d]\comp [-f,-e]=[c-f,d-e]$, we show
  % that, for any $x\in [a,b]$, there exists $y\in [c,d]$ such that
  % $y-x \in [e,f]$, i.e., $x+e \leq y \leq x+f$. If $x+e\leq c$, then
  % we take $y=c$; if $x+f\geq d$, then we take $y=d$; if $x+e>c$ and
  % $x+f<d$, then $c-e <x < d-f \leq d-e$, and, hence, $c< x+e < d$.
  % Take $y=x+e$. It is easy to check that $x+e \leq y \leq x+f$ is
  % satisfied in all \magenta{three} cases.
\end{proof}

\begin{lemma}\label{lemma:AC}
  Let $\network$ be an arc-consistent STN and $v,w \in V$ two
  variables that are constrained by $I_{vw}$ in $C$. Then
  $I_v\subseteq I_w\comp I_{wv}$.
\end{lemma}
\begin{proof}
  This follows directly from Lemma~\ref{lemma:AC-0} and that $I_{vw}$
  is AC relative to $I_v$ and $I_w$.
  % Write $I_v=[a,b]$, $I_v=[c,d]$, and $I_{vw}=[e,f]$. Since $I_{ji}$
  % is AC, for any $y\in [c,d]$, there exists $x\in [a,b]$ such that
  % $y-x\in [e,f]$, i.e., $e\leq y-x\leq f$. We have thus
  % $x+e\leq y\leq x+f$, which is equivalent to saying that
  % $a+e\leq y\leq b+f$, i.e.,
  % $y\in [a,b]\comp [e,f]=I_v\comp I_{vw}$.
\end{proof}

% \begin{lemma}\label{lemma:AC}
%   Let $\network$ be an STN and $(v, w) \in E$. Then $I_{vw}$ is AC
%   if and only if $I_v\comp I_{vw} \subseteq I_w$.
% \end{lemma}

% \begin{proof}
%   Suppose $I_{vw}$ is AC and write $I_v=[a, b]$, $I_w=[c, d]$, and
%   $I_{vw}=[e, f]$. Since $I_{vw}$ is AC, for any $x \in [a, b]$,
%   there exists $y \in [c, d]$ such that $y - x \in [e,f]$, i.e.,
%   $e \le y - x \le f$. We have thus $x + e \le y \le x + f$, which
%   is equivalent to saying that $a + e \le y \le b + f$, i.e.,
%   $y\in [a,b]\comp [e,f]=I_v\comp I_{vw}$.
% \end{proof}

The following result directly follows from Lemma~\ref{lemma:AC}.

\begin{corollary}\label{cor:path}
  Let $\network$ be an arc-consistent STN. Let $\pi$ be a path in $\N$
  from $w$ to $v$. Then $I_v \subseteq I_w \comp \bigotimes \pi$.

  % Let $\pi=v_{i_0},v_{i_1},...,v_{i_k}$ be a path in $\N$ from $v$
  % to $w$. Write $I_x$ for the domain of $v_x$ and $J_{s,s+1}$ for
  % the constraint in $C$ between $v_{i_s}$ and $v_{i_{s+1}}$. Then
  % $I_v \subseteq I_v\comp J_{0,1} \comp J_{1,2} \comp \ldots \comp
  % J_{k-1,k} = I_v \comp \bigotimes \pi$.
\end{corollary}

\begin{lemma}\label{lemma:minimal}
  Let $\network$ be an arc-consistent STN and $v,w$ variables in $V$.
  If $\N$ is consistent, then $I_v\subseteq I_w\otimes I^m_{wv}$, where
  $I^m_{wv}$ is the minimal constraint from $w$ to $v$.
\end{lemma}
\begin{proof}
  Since $\N$ is consistent, $I^m_{wv}$ is nonempty. Recall that
  $I^m_{wv}$ is the intersection of the compositions along all paths
  in $\mathcal{N}$ from $w$ to $v$ (cf.~\cite[\S 3]{DechterMP91}) and
  composition distributes over non-empty intersection for intervals.
  The result follows directly from Corollary~\ref{cor:path}.
\end{proof}

 \begin{lemma}[\cite{Shostak:1981:DLI:322276.322288}] \label{lem:inconsisent_stn}
   Suppose $\network$ is an STN. Then $\N$ is inconsistent if and only
   if there exists a negative cycle.
 \end{lemma}
 \begin{lemma}\label{lem:2}
   \red{Given a consistent STN $\network$ with $n = |V|$, for any path
   $\pi$ of length $\ge n$ there is a path $\pi'$ of length $< n$ such
   that $\bigotimes \pi' \subseteq \bigotimes \pi$.}
 \end{lemma}
 \begin{proof}
   \red{Since the length of $\pi$ is $\ge n$, $\pi$ must have a cycle at
     a variable $v$. As the cycle is not negative, removing the cycle
     and leaving only $v$ in the path results in a path $\pi'$ with
     $\bigotimes \pi' \subseteq \bigotimes \pi$. Repeating this
     procedure until there is no cycle gives the desired result.}
 \end{proof}

\begin{lemma}\label{lem:1}
  Let $\network$ be an STN and $\N'$ its AC-closure. Then $\N$ is
  consistent iff $\N'$ has no empty domain.
\end{lemma}
\begin{proof}
  We prove $\N$ is inconsistent iff $\N'$ has an empty domain. As $\N$
  and $\N'$ are equivalent, if $\N'$ has an empty domain, then $\N$ is
  inconsistent.

  Now suppose $\N$ is inconsistent. \red{Then by
    Lemma~\ref{lem:inconsisent_stn}, there exists a negative cycle
    $\pi$ in $\N$ at some $w$ such that $\bigotimes \pi=[l, h]$ with $h < 0$. Now let $v$ be a variable in
    $\N$ with $I_{wv} = [e, f]$ and let $I_v'=[a, b], I_w'=[c, d]$ be
    the domains of $v$ and $w$ in $\N'$, respectively. Choose
    $k \in \mathbb{N}$ sufficiently large, such that $kh < b - d - f$.
    Then, by Lemma~\ref{lemma:minimal} we have
    \begin{IEEEeqnarray}{rCl}
      I_v' &\subseteq& I_w' \comp \left(\bigotimes \pi^k \comp I_{wv}\right)\label{eq:2}\\
      &=& [c, d] \comp ([kl, kh] \comp [e, f])\nonumber\\
      &=& [c + kl + e, d + kh + f]\nonumber,
    \end{IEEEeqnarray}
    where $\pi^k$ is the concatenation of $k$ copies of path $\pi$.
    Because $kh < b - d - f$, \eqref{eq:2} is possible only if $I_v'$
    is empty.%
  }%
  % We prove $\N$ is inconsistent iff $\N'$ has an empty domain. As $\N$
  % and $\N'$ are equivalent, if $\N'$ has an empty domain, then $\N$ is
  % inconsistent.
  % Now suppose $\N$ is inconsistent. Then by
  % Lemma~\ref{lem:inconsisent_stn}, there exists a negative cycle $\pi$
  % in $\N$ at some $v$ such that $\bigotimes \pi=[a, b]$ with $b < 0$.
  % Let $I_v=[l, h]$ \magenta{be} \magenta{the domain of $v$} in $\N'$.
  % Then by Lemma~\ref{lemma:minimal},
  % $[l, h] \subseteq [l, h] \comp \bigotimes \pi = [l, h] \comp [a,
  % b]=[l + a, h + b]$. Since $h+b < h$, this is possible iff $I_v$ is
  % empty.
\end{proof}
\begin{theorem}\label{the:minimal}
  Let $\network$ be a consistent STN and $\N'$ its AC-closure. Then
  all domains in $\N'$ are minimal.
\end{theorem}
\begin{proof}
  If the constraint graph $G_{\N}$ is connected, i.e., for any two
  variables $v,w$, there is a path in $G_{\N}$ that connects $v$ to
  $w$, then we may replace the constraint from $v$ to $w$ with the
  nonempty minimal constraint $I^m_{vw}$ \red{(or add $I^m_{vw}$, if
    there was no constraint between $v$ and $w$)}. We write the
  refined STN as $\mathcal{N}^*$. For any two variables $v,w$, by
  Lemma~\ref{lemma:minimal}, $I_v$ is contained in
  $I_w \comp I^m_{wv}$ and $I_w$ is contained in $I_v \comp I^m_{vw}$.
  This shows that $\mathcal{N}^*$ is the same as the minimal STN of
  $\N$, and thus, establishes the minimality of each $I_v$.

  In case the constraint graph is disconnected, we consider the
  restriction of $\N$ to its connected components instead. The same
  result applies.
\end{proof}

Two special solutions can be constructed if $\N$ is arc-consistent and
has no empty domain.

\begin{proposition}\label{prop:solution}
  Let $\network$ be an arc-consistent STN with $D=\{I_v \mid v\in V\}$
  and $I_v=[a_v,b_v]$ for each $v$. If no $I_v$ is empty, then the
  assignments $A = \Set*{a_{v}\given v\in V}$ and
  $B = \Set*{b_{v} \given v \in V}$ are two solutions of
  $\mathcal{N}$.
\end{proposition}
\begin{proof}
  Let $\N'=\la V, D', C'\ra$ be the minimal STN of $\N$. By
  Theorem~\ref{the:minimal}, we have $D'=D$ and $\N'$ is equivalent to
  $\N$. The above claim follows as the assignments
  $A = \Set*{a_{v}\given v\in V}$ and
  $B = \Set*{b_{v} \given v \in V}$ are two solutions of the minimal
  STN $\N'$ (cf. \cite[Corollary 3.2]{dechter2003constraint}).
\end{proof}

\begin{theorem}\label{thr:1}
  Enforcing AC is sufficient to solve STP.
\end{theorem}
\begin{proof}
  Let $\N$ be an STN and $\N'$ its AC-closure. If $\N'$ has an empty
  domain, then $\N$ has no solution by Lemma~\ref{lem:1}. If $\N'$
  does not have an empty domain, then we can use
  Proposition~\ref{prop:solution} to find a solution.
\end{proof}
\begin{remark}\label{rem:2}
  (i) \red{As solving an STN is equivalent to solving a system of
    linear inequalities, the solution set of an STN is a convex
    polyhedron.} Thus any convex combination of the two solutions $A$
  and $B$ is again a solution of the STN. (ii) Enforcing AC can in
  essence find all solutions of an STN: Suppose $\N$ is arc-consistent
  and has no empty domain. We pick an arbitrary variable $v$ that has
  not been instantiated yet, then assign any value from $D_v$ to $v$,
  and enforce AC on the resulting network. We repeat this
  process until all variables are instantiated. \red{(iii)
    Proposition~\ref{prop:solution} can also be obtained by first
    showing that STP constraints are both max/min-closed, and then
    using the result in \cite[Thm~4.2]{jeavons_tractable_1995}, which
    states that the AC-closure of a constraint network over
    max/min-closed constraints have the maximal and the minimal values
    of the domains as two solutions. As a consequence of this,
    Theorem~\ref{the:minimal} can also be obtained, because the
    solution set of an STN is convex (cf.~Remark~\ref{rem:2}~(i)).}
\end{remark}

\subsection{A Centralized AC Algorithm for the STP}
In this section we propose an AC algorithm, called \AC, to solve STNs.
The algorithm is presented as Algorithm~\ref{algorithm_ac}.

\begin{algorithm}[t]
  \DontPrintSemicolon \SetKwInOut{Inpu}{Input}
  \SetKwInOut{Output}{Output} \SetKw{KwSt}{s.t.} \Inpu{An STN
    $\network$ and its constraint graph $\graph$, where
    \red{$|V| = n$}.} \Output{An equivalent network that is AC, or
    ``inconsistent''.} \BlankLine $Q \gets \varnothing$\;
  % \magenta{$\textsc{change} \gets {\false}$}\;
  \For{$k \gets 1$ \KwTo {$n$}}{\label{ln:for-start}
    % \magenta{$Q' \gets Q$}\;
    \ForEach{$v \in V$}{ $\magenta{I_v' \gets I_v}$\;%
      \ForEach{$w \in V$ \KwSt $\Set*{v,w} \in E$}{
        $I_v \gets I_v \cap I_w \comp I_{wv}$\label{ln:update}\; }
      \lIf{$I_v = \varnothing$}{ \Return ``inconsistent'' }

      \lIf{$I_v' = I_v$}{\label{ln:equal} $Q \gets Q \cup \{v\}$ }
      \lElse{ $Q \gets Q \setminus \{v\}$ } } \lIf{$\#Q = n$}{ \Return
      $\mathcal{N}$ \label{ln:break} }
    % \magenta{ \lIf{$Q=Q'$} {$\textsc{change} \gets \false$} \lElse{
    % $\textsc{change} \gets \true$} }
  }\label{ln:for-end}
  % \magenta{\lIf{$\#Q < n$}
  \Return ``inconsistent'' \label{ln:overlimit}
  % { \lElse{\Return $\mathcal{N}$ \label{ln:exit} } } \JL{Shall the
  % algorithm return a solution or a network?}
  \caption{\AC}\label{algorithm_ac}
\end{algorithm}

\begin{theorem}\label{thm:centralized_AC}
  Given an input STN $\mathcal{N}$, Algorithm~\ref{algorithm_ac}
  returns ``inconsistent" if $\N$ is inconsistent. Otherwise, it
  returns the AC-closure of $\N$.
\end{theorem}
\begin{proof}
  We first note that intersection and composition of constraints do
  not change the solution set of the input STN $\mathcal{N}$. This has
  two implications: First, if a domain $I_v$ becomes empty during the
  process of the algorithm, then the solution set of $\mathcal{N}$ is
  empty and $\mathcal{N}$ is inconsistent. Second, if the algorithm
  terminates and its output $\mathcal{N}'$ is AC, then $\mathcal{N}'$
  is the AC-closure of $\mathcal{N}$. \red{Consequently, it suffices
    to show that if the algorithm terminates and returns
    $\mathcal{N}'$, then $\mathcal{N}'$ is AC.}

  % Algorithm~\ref{algorithm_ac} exits either in line~\ref{ln:break}
  % or in line~\ref{ln:exit}.

  We first consider the case, where the algorithm returns $\N'$ in
  line~\ref{ln:break} at the $k$th iteration of the for-loop
  (lines~\ref{ln:for-start}--\ref{ln:for-end}) for some
  $1\leq k\leq n$. We show that $\mathcal{N}'$ is AC. Let $I_v^k$ be
  the domain of $v$ obtained after the $k$th iteration of the
  for-loop. Due to lines~\ref{ln:update} and \ref{ln:equal}, we have
  for all $\Set*{v,w} \in E$ that
  $I_v^k \subseteq I_v^{k-1} \cap (I_w^{k-1} \otimes I_{wv})$ and
  $I_w^{k-1} = I_w^k$. Thus we have for $\Set*{v,w} \in E$ that
  $I_v^k \subseteq I_w^{k} \otimes I_{wv}$, which is by
  Lemma~\ref{lemma:AC-0} equivalent to saying that $I_{vw}$ is AC
  w.r.t.\ domains $I_v^k$ and $I_w^k$. Hence, the output $\mathcal{N}'$
  is AC.

  Now suppose that the algorithm exited in line~\ref{ln:overlimit}
  returning ``inconsistent''. Thus, at the $n$th iteration of the
  for-loop we have $\#Q < n$ in line~\ref{ln:break}. We prove that
  $\N$ is inconsistent by contradiction. Assume that $\N$ is
  consistent.
  % We first claim that $I_v^k$ is a intersection of paths
  % from $o$ (the auxiliary variable denoting the zero time point) to
  % $v$ of lengths $\ge k$. We prove this by induction on $k$. The claim
  % is trivially true for $k=0$, as $I_v^0 = I_v$ is the intersection of
  % path $\pi = (o, v)$
  For any $v\in V$ and any $k\geq 1$, we write $\Pi_{v}^k$
  for the set of paths from $o$ (the auxiliary variable denoting the
  zero time point) to $v$ with length $\leq k$ in the constraint graph
  of $\N$. We claim

  \begin{equation}
    I_v^{k-1} \subseteq \bigcap_{\pi \in \Pi^{k}_v }\bigotimes{ \pi} \label{eq:5}
  \end{equation}
  for any $k \geq 1$. \red{Then, with $I_v^m $ being the minimal
    domain of $v$, we have
    \begin{equation*}
      I_v^m \subseteq I_v^{n-1} \subseteq \bigcap_{\pi \in \Pi^{n}_v }\bigotimes{ \pi} = I_v^m,
    \end{equation*}
    because $I^m_{v}$ is the intersection of the compositions along
    all paths in $\mathcal{N}$ from $o$ to $v$ (cf.~\cite[\S
    3]{DechterMP91}), where it suffices to only build compositions
    along paths of length $\le n$ by Lemma~\ref{lem:2}. Thus
    $I_v^{n} = I_v^{n-1} = I_v^m$ for all $v \in V$, which is a
    contradiction to our assumption that at the $n$th iteration of the
    for-loop we have $\#Q < n$ in line~\ref{ln:break}.%
  }
  
  \red{%
   We now prove \eqref{eq:5} by using induction on $k$.
  First, for $k = 1$, since $\Pi_v^1$ contains only one path of length
  1 (i.e., the edge $\{o,v\}$), we have
  $I_v^0 = I_v = \bigcap_{\pi\in \Pi_v^1 } \bigotimes \pi$. Now
  suppose \eqref{eq:5} is true for $k-2$ for all $w \in V$. Then by
  line~\ref{ln:update} and our induction hypothesis we have}
  \begin{align*}
    I_v^{k-1} & \subseteq I_v^{k-2} \cap \left(\bigcap_w I_w^{k-2}  \otimes I_{wv}\right)                                                  \\
             & \subseteq I_v^{k-2} \cap \left(\bigcap_w \left(\bigcap_{\pi \in \Pi_{w}^{k-1}}\bigotimes{\pi}\right) \otimes I_{wv}\right) \\
             & \subseteq \left( \bigcap_{\pi \in \Pi_{v}^{k-1}}\bigotimes{\pi} \right) \cap \left( \bigcap_{(\pi \in \Pi_{v}^{k}) \wedge (|\pi|\geq 2)}\bigotimes{\pi}  \right)                                                \\
             & = \bigcap_{\pi\in \Pi_{v}^{k}}\bigotimes {\pi},
  \end{align*}
  which proves \eqref{eq:5}.%
  % We now claim that $I_v^{n-1}=I_v^n$ for every $v\in V$. Note that
  % every path $\pi$ of length \blue{$>n$} contains a cycle, say
  % $\theta$. Write $\pi^*$ for the path with $\theta$ removed. If $\N$
  % is consistent, then it is routine to show that
  % $\bigotimes \pi^* \subseteq \bigotimes \pi$ \blue{(Can anyone say
  %   something here to explain this?)}. This implies that, by
  % \eqref{eq:5}, $I_v^{n-1}=I_v^n$ holds for every $v\in V$, which
  % contradicts the assumption that $\#Q<n$. Therefore, $\mathcal{N}$ is
  % inconsistent.
\end{proof}

\begin{theorem}
  Algorithm~\ref{algorithm_ac} runs in time $O(en)$, where $e$
  is the number of edges of the constraint graph of the input STN and
  $n$ is the number of variables.
\end{theorem}
\begin{proof}
  There are at most $n$ iterations of the for-loop and each iteration
  involves $O(e)$ operations.
\end{proof}

\begin{remark}\label{rem:3}
  \red{Algorithm~\ref{algorithm_ac} can also be understood as
    computing the shortest path from a source vertex $o$ to every
    other vertex $v$ and the shortest path from every other vertex $v$
    to the source vertex $o$. This can be realized in time $O(en)$ by
    using a shortest path tree algorithm with negative cycle
    detection~(cf.~\cite[Section~7.2]{tarjan_data_1983} and
    \cite[Section~7.1]{korte_combinatorial_2012}.}
\end{remark}

\section{Solving the MaSTP with Arc-Consistency}\label{sec:solving-mastp-with}
In this section we extend \AC to a distributed algorithm \DAC to solve
multiagent simple temporal networks (MaSTNs).

% In the following, we assume that $\Lambda$ is a set of agents.

\begin{definition}\cite{DBLP:journals/jair/BoerkoelD13}\label{def_dis}
  A \emph{multiagent simple temporal network (MaSTN)}
  is a tuple
  ${\mathcal{M}}=\la{\cal P}, C^\text{X} \ra$, where
  \begin{itemize}
  \item ${\cal P}= \Set*{\N_i \given i = 1, \dotsc, p}$ is a set of
    \emph{local} STNs, where each
    $\mathcal{N}_{{i}}=\la V_{{i}}, {D_i}, C_{{i}}
    \ra$ is an STN belonging to agent $i$ and we require that
    $V_i\cap V_j=\varnothing$ for any two different agents
    $i,j = 1, \dotsc, p$.
  \item $C^\text{X}$ is a set of
    \emph{external} constraints, where each constraint is over two
    variables belonging to two different agents.
  \end{itemize}
\end{definition}

\noindent Constraint graphs for MaSTNs can be defined analogously as
that for STNs, where we use $E^\text{X}$ for the set of edges
corresponding to constraints in $C^{X}$. See Figure~\ref{fig:mastp}
for an illustration. In Figure~\ref{fig:mastp}, the edges in
$E^\text{X}$ are represented as red lines.

\begin{definition}
  Suppose
  ${\mathcal{M}}=\la {\cal P}, C^\text{X}\ra$ is an MaSTN. Let
  $I_{vw} \in C^\text{X}$ with
  $v \in V_i, w \in V_j$ be an external constraint.
  We say that $I_{vw}$ is an external constraint of agent $i$, and
  write $C_i^\text{X}$ for the set of external constraints of agent
  $i$. We call
  $v$ and $w$ a \emph{shared} and an
  \emph{external} variable of
  agent $i$, respectively. We write $V_i^\text{X}$ for the set of
  external
  variables of agent $i$. In Figure~\ref{fig:mastp}, the vertices for
  shared variables are represented as red circles.
\end{definition}
%\noindent In Figure~\ref{fig:mastp}, shared (external) variables are
%marked as light green nodes.

\begin{algorithm}
  % \small
  \DontPrintSemicolon%
  \SetKwInOut{Inpu}{Input}%
  \SetKwInOut{Output}{Output}%
  \SetKw{KwSt}{s.t.}%
  \Inpu{%
    $\N_i$: agent ${i}$'s portion of MaSTN $\mathcal{M}$; \newline%
    $V_i^\text{X}$: the set of agent $i$'s external variables;
    \newline%
    $C_i^\text{X}$: the set of agent $i$'s external
    constraints;\newline%
    $\textit{parent}(i)$: the parent of agent $i$ w.r.t.\
    $T(\mathcal{M})$;\newline%
    $\textit{children}(i)$: the children of agent\ $i$ w.r.t.\
    $T(\mathcal{M})$;\newline%
    $n$: the number of variables of $\M$. }%
  \Output{%
    Agent $i$'s portion of the AC-closure of $\mathcal{M}$ or
    ``inconsistent''.%
  }%
  $Q_i \gets \varnothing$\;
  % Construct a minimal spanning tree of the agent communication
  % network\;%
  \For{$k \gets 1$ \KwTo $n$\label{ln:for1}}{%
    % \smallskip%
    Send the domains of the shared variables to the
    neighbors.\label{ln:sync1}\;%
    Receive the domains of the external variables from the neighbors.
    \label{ln:sync2}\;%
    \ForEach{$v \in V_i$\label{ln:update1}}{%
      $I_v' \gets I_v$\;%
      \ForEach{$w \in V_i \cup V_i^{\text{X}}$ \KwSt
        $\{v, w\} \in E\magenta{_i} \cup E_i^\text{X}$}{%
        $I_v \gets I_v \cap I_w \comp
        I_{wv}$\label{ln:update-mastn}\;%
      }%
      \If{$I_v = \varnothing$\label{ln:inconst1}}{%
        Broadcast ``inconsistent''.\;%
        \Return ``inconsistent''
      }\label{ln:inconst2}%
      \lIf{$I_v' = I_v$}{ $Q_i \gets Q_i \cup \{v\}$ } \lElse{
        $Q_i \gets Q_i \setminus \{v\}$ } }
    \If{$\#Q_i = \#V_i$\label{ln:qfull1}}{%
      \If{$\textit{root}(i)$\label{ln:root1}}{%
        Send inquiry (``Are all $Q_i$ full?'', $k$) to
        $\textit{children}(i)$%
      }\label{ln:root2}%
      \While{true}{%
        $m \gets \textsc{ReceiveMessage}()$\;
        \If{$m$ is domains of external variables from a
          neighbor\label{ln:continue1}}{%
          \textbf{break}%
        }\label{ln:continue2}%
        \If{$m$ is inquiry (\textnormal{``Are all $Q_i$ full?'',
            $k$)}\label{ln:inquiry1}}{%
          \If{$\textit{leaf}(i)$\label{ln:leaf1}}{%
            Send feedback (``yes'', $k$) to $\textit{parent}(i)$%
          }\label{ln:leaf2}%
          \lElse{%
            Send $m$ to $\textit{children}(i)$%
          }
        }%
        % \If{$m$ is a feedback \textnormal{(``no'', $k$)}}{%
        % Forward $m$ to $\textit{parent}(i)$\;
        % \textbf{break}%
        % }%
        \If{$m$ is \textit{feedback} \textnormal{(``yes'',
            $k$)}\label{ln:feedback1}}{%
          \If{all feedbacks received from $\textit{children}(i)$}{%
            \If{$\textit{root}(i)$}{%
              Broadcast ``arc-consistent''\;
              \Return $\N_i$%
            }%
            \lElse{Send $m$ to $\textit{parent}(i)$%
            }%
          }%
        }\label{ln:feedback2}\label{ln:inquiry2}%
        \If{$m$ is \textnormal{``arc-consistent''}}{%
          \Return $\N_i$%
        }%
        \If{$m$ is \textnormal{``inconsistent''}}{%
          \Return ``inconsistent''%
        }%
      }%
    }\label{ln:qfull2}%
    % Send feedback (``no'', $k$) to $\textit{parent}(i)$\;%
  }\label{ln:for2}%
  \Return ``inconsistent''
  \caption{\DAC}\label{algorithm_dac}
\end{algorithm}

% \SL{How many messages are sent and received by agent $i$?}
% \SL{In each round, if $root(i)$, then $e^X_i$ (line 3) $+e^X_i$ (line 4) $+deg(i)$ (line 16) $+deg(i)$ (line 26) messages are received or sent by agent $i$, where $e^X_i$ is the number of external edges in $\N_i$ and $deg(i)$ is the number of children of agent i. In total, we have at most $n\times$ the above number and $+$ the number of edges in the spanning tree (for line 10 and line 28). The other agents are similar.  }

\DAC is presented in Algorithm~\ref{algorithm_dac}. In \DAC each agent
$i$ gets as input its portion $\N_i$ of the input MaSTN $\mathcal{M}$
and the set $C_i^\text{X}$ of its external constraints, and runs its
own algorithm. Similar to \AC, \DAC updates the domains of $\N_i$ at
each iteration of the for-loop and maintains a queue $Q_i$ to record
the information about the unchanged domains. When a domain becomes
empty during the updates, then the agent can terminate the algorithm
and conclude that the input MaSTN $\mathcal{M}$ is inconsistent. There
are however aspects in \DAC that are different from \AC, which stem
from the fact that in MaSTP an agent cannot have the global knowledge
of the states of other agents' processes without sharing certain
information with other agents. These aspects are the following:

\begin{enumerate}

\item The total number $n$ of the variables in the input MaSTN is
  initially not known to individual agents. This, however, can easily
  be determined using an echo algorithm~\cite{chang_echo_1982}. We can
  therefore regard $n$ as given as an input to \DAC.

\item As the agents may run their processes at different paces, at
  each iteration of the for-loop (lines~\ref{ln:for1}--\ref{ln:for2}),
  they synchronize the domains of their external variables
  (lines~\ref{ln:sync1}--\ref{ln:sync2}). Otherwise, some agents might
  use stale external domains and make wrong conclusions.
  
\item When a domain becomes empty while running \DAC, an agent
  broadcasts (lines~\ref{ln:inconst1}--\ref{ln:inconst2}) this
  information to other agents so that they can
  terminate their algorithms as
  soon as possible.
  
\item If the queue $Q_i$ of an agent $i$ is full (i.e., it
  contains all of the agent's variables in $V_i$) after an
  iteration of the for-loop, then the
  agent shares this information with all other agents in
  $\mathcal{M}$ so as to jointly determine whether the queues of all
  agents are full and the network is arc-consistent
  (lines~\ref{ln:root1}--\ref{ln:root2} and
  \ref{ln:inquiry1}--\ref{ln:inquiry2}).
  
\item If the queue $Q$ of an agent is \emph{not} full after an
  iteration of the for-loop, then the agent
  broadcasts
  this information to all other agents, so that they can move to the
  next iteration of the for-loop as soon as possible.

\end{enumerate}
All the preceding aspects are subject to communication of certain
information between agents. \DAC coordinates this communication while
(i) preserving the privacy of each agent and (ii) reducing the
duration of any idle state of an individual agent. Concretely:

\begin{itemize}
\item Each agent shares information only with the agents who are
  connected through an external constraint. We call them the
  \emph{neighbors} of the agent. This neighborhood-relationship among the
  agents
  induces a graph that we call henceforth the \emph{agent
    graph}.

\item Each agent shares with its neighbors \emph{only} the
  domains of its
  shared
  variables. No other information is shared (such as its network
  structure, constraints, private variables and their domains) and
  only the neighbors w.r.t.\ the agent graph can share the information.
  This property is a critical advantage over \DPPC
  \cite{DBLP:journals/jair/BoerkoelD13}, as \DPPC often creates new
  external constraints during the process and reveal more private
  information of the agents than necessary.

\item Each agent uses a \emph{broadcasting} mechanism to share global
  properties of the input MaSTN, i.e., an agent first sends a message
  (e.g., ``inconsistent'') to its neighbors, then the neighbors
  forward the message to their neighbors and so on, until all agents
  receive the message. To reduce the number of messages, duplicates
  are ignored by the agents.

  An agent $i$ broadcasts the following messages: ``arc-consistent'',
  ``inconsistent'' and ``$Q_i$ is not full'', where the last message
  is indirectly broadcasted by agent $i$ skipping
  lines~\ref{ln:qfull1}--\ref{ln:qfull2} and moving to the next
  iteration of the for-loop and then sending its shared domains to its
  neighbors. This initiates a chain reaction among the idle neighbors
  of agent $i$ who have not moved to the next iteration yet, as they
  quit the idle states (lines~\ref{ln:continue1}--\ref{ln:continue2})
  and move to the next iteration of the for-loop and then send also
  their shared domains to their idle neighbors
  (lines~\ref{ln:sync1}--\ref{ln:sync2}).
  
\item There is a dedicated agent who checks at each iteration of its
  for-loop (given its queue is full) whether the queues of all other
  agents are full at the same iteration. This dedicated agent is
  determined by building a minimal spanning tree (e.g., by using an
  echo algorithm~\cite{chang_echo_1982}) $T(\mathcal{M})$ of the agent
  graph. The agent who is the root (henceforth the \emph{root agent})
  of this tree becomes then the dedicated agent.

  The root agent sends an inquiry to its
  children to check whether the queues of all its descendants are
  full (lines~\ref{ln:root1}--\ref{ln:root2}). The inquiry is then
  successively forwarded by the descendants
  whose queues are full. We have to distinguish here between two
  cases:

  (1) If all descendants' queues are \emph{full}, then the
  inquiry reaches all the leaf agents and returns back as feedbacks
  (lines~\ref{ln:leaf1}--\ref{ln:leaf2})
  until the root agent receives all the feedbacks
  (lines~\ref{ln:feedback1}--\ref{ln:feedback2}) and broadcasts
  ``arc-consistency''.

  (2) If a descendant's queue is \emph{not full},
  then the descendant moves on to the next iteration of the for-loop
  and initiates a chain reaction among other agents by sending the
  domains of its shared variables to its neighbors (cf.~the second
  paragraph of the third bullet point).

\end{itemize}

Due to the properties so far considered, \DAC is guaranteed to
simulate the behavior of \AC while allowing concurrent domain update
operations.

\begin{theorem}
  % Let $d$ be the radius of the constraint graph of $\M$ and $n$ the
  % number of variables in $\M$.
  Let $\M = \la \mathcal{P}, C^\text{X} \ra$ be an MaSTN. Let
  $\N_\textnormal{max}$ be a network with
  $e_\textnormal{max}= \max\Set*{e_i + e_i^\textnormal{X}\given 1 \le
    i \le p}$, where $e_i$ and $e_i^\textnormal{X}$ are the number of
  edges of the constraint graph of $\N_i$ and the number of external
  constraints of agent $i$, respectively. Then
  Algorithm~\ref{algorithm_dac} enforces AC on $\M$ in time
  $O(e_\textnormal{max}n)$.
\end{theorem}

\begin{figure*}[t]
  \centering
  \hfill
\subcaptionbox{\textsf{Scale-free-1}}{
\begin{tikzpicture}
\begin{axis}[
%    %title={(a) Performance evaluation of the two algorithms in the number $n$ of variables. We set $\rho=0.5$ and $d=20$.},
    ymode = log,
%    %log ticks with fixed point,
    set layers=standard,
    xlabel={Network density},
    xmin=5, xmax=50,
    ymin=0, ymax=500000000,
    xtick={5,10,20,30,40,50},
%
%    %ytick={0, 100000000,200000000,300000000,400000000, 500000000},
%
%    %ytick={0, 100000,200000,300000,400000, 500000,600000,700000},
%    %legend pos=north west,
    legend style={at={(0.5,0.6)},anchor=west, font=\scriptsize, draw=none},
%    %, fill opacity=0.1, text opacity = 1
    ymajorgrids=true,
    xmajorgrids=true,
    grid style=dashed
]

\addplot[
    color=blue,
    mark=square,
    on layer={axis foreground}
    ]
    coordinates {
    (8,128511621)(14,242283360)(20,298364793)(26,336141948)(35, 377131148)(50, 436242948)
    };
    %\legend{CuSO$_4\cdot$5H$_2$O}
    \addlegendentry{\PPC}
    
    \addplot[
    color=red,
    mark=triangle,
    ]
    coordinates {
    (8,141984)(14,190904)(20,307840)(26,395200)(35, 473100)(50, 593300)
    };
      \addlegendentry{\AC}
    
\end{axis}
\end{tikzpicture}}
\hfill
\subcaptionbox{\textsf{Scale-free-2}}[.3\linewidth]{
\begin{tikzpicture}
\begin{axis}[
%    %title={(a) Performance evaluation of the two algorithms in the number $n$ of variables. We set $\rho=0.5$ and $d=20$.},
    ymode = log,
    xlabel={$n$},
    xmin=300, xmax=800,
    xtick={300,400,500,600,700,800},
%    %ytick={0,4000000,8000000,12000000,16000000,20000000},
%    legend pos=north west,
    legend style={at={(0.5,0.6)},anchor=west, font=\scriptsize, draw=none},
    ymajorgrids=true,
    xmajorgrids=true,
    grid style=dashed,
]

\addplot[
    color=blue,
    mark=square,
    ]
    coordinates {
    (300,4323192)(400,7079856)(500,8316399)(600,12038190)(700,16038190)(800,17297700)
    };
    %\legend{CuSO$_4\cdot$5H$_2$O}
    \addlegendentry{\PPC}
    
    \addplot[
    color=red,
    mark=triangle,
    ]
    coordinates {
    (300,30784)(400,32760)(500,34440)(600,35520)(700,38864)(800,50816)
    };
      \addlegendentry{\AC}
    
\end{axis}
\end{tikzpicture}}
\hfill
\subcaptionbox{\textsf{New York}}[.3\linewidth]{
\begin{tikzpicture}
\begin{axis}[
%    %title={(a) Performance evaluation of the two algorithms in the number $n$ of variables. We set $\rho=0.5$ and $d=20$.},
    ymode = log,
    xlabel={$n$},
%    %ylabel={Estimated CPU time (sec)},
    xmin=300, xmax=1024,
    xtick={200,400,600,800, 1024},
%
%    % ytick={0,60000,120000,180000,240000,300000},
%
%    legend pos=north west,
    legend style={at={(0.5,0.6)},anchor=west, font=\scriptsize, draw=none},
    ymajorgrids=true,
    xmajorgrids=true,
    grid style=dashed
]

\addplot[
    color=blue,
    mark=square,
    ]
    coordinates {
    (335,22689)(419,29037)(524,139299)(655,167358)(819,191109)(1024,252372)
    };
    %\legend{CuSO$_4\cdot$5H$_2$O}
    \addlegendentry{\PPC}
    
    \addplot[
    color=red,
    mark=triangle,
    ]
    coordinates {
    (335,4490)(419,5710)(524,7140)(655,10812)(819,16058)(1024,20300)
    };
      \addlegendentry{\AC}
    
\end{axis}
\end{tikzpicture}
}
\hfill{}

 % \JL{Redraw the graphs using pgfplot? Included the variables for parameters.}
\caption{Evaluation of \AC and \PPC. The $y$-axes (on the log scale)
  represent the number constraint checks.}\label{fig:exp_cen}
\end{figure*}
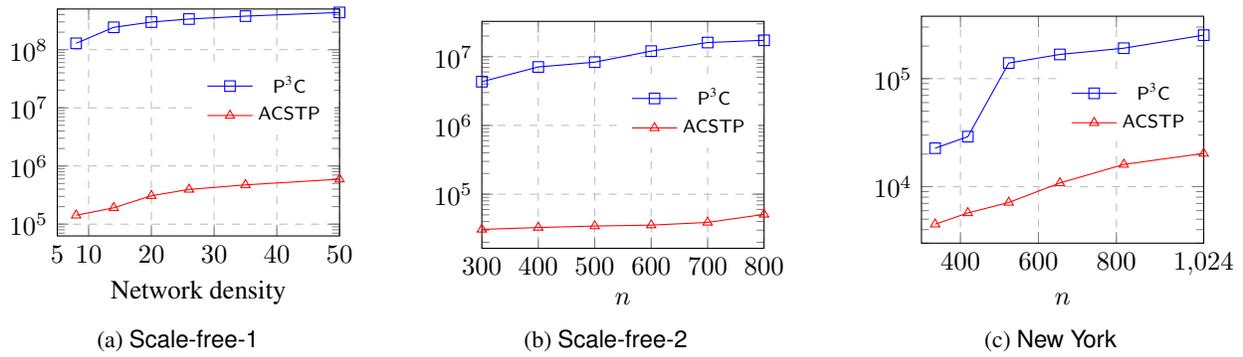

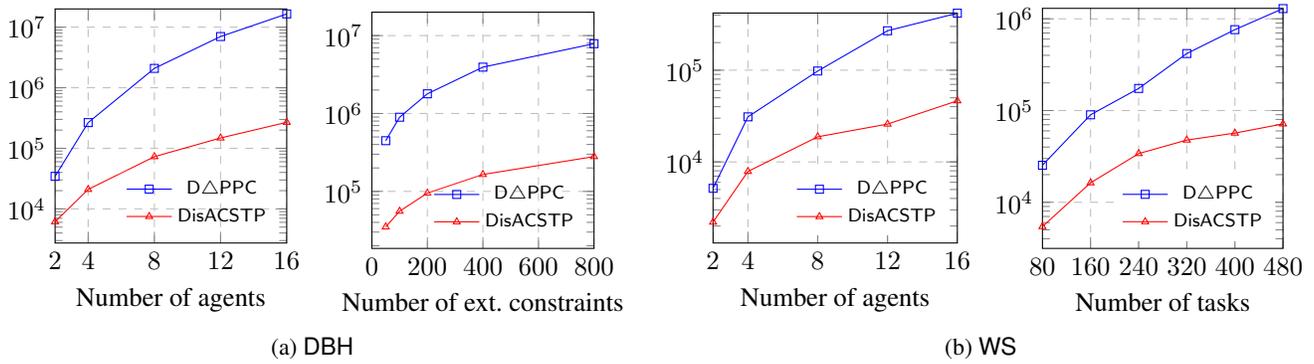
\begin{figure*}
  \begin{subcaptionbox}{\textsf{DBH}\label{fig:dbh}}[.48\linewidth]{
        \begin{tikzpicture}
      \begin{axis}[xlabel={Number of agents},
        % width=.24\linewidth,
        xmin=2, xmax=16, ymin=0, ymax=20000000,
        xtick={2, 4, 8, 12, 16},
        x=0.22cm, y=0.35cm,
        legend pos=south east, 
        legend style={font=\scriptsize, draw=none, fill opacity=0.1, text opacity = 1},
        label style={font=\normalsize}, tick label
        style={font=\normalsize},
        ymajorgrids=true,
        xmajorgrids=true,
        grid style=dashed,
        ymode = log
        ]
       
        \addplot[ mark size=1.5pt, color=blue, mark=square, ]
        coordinates {
         (2,34440)(4,265680)(8,2086560)(12,6998640)(16,16537920)
        }; \addlegendentry{\DPPC}
        
         \addplot[ mark size=1.5pt, color=red, mark=triangle, ] coordinates
        {
          (2,6140)(4,20980)(8,72780)(12,148140)(16,268680)
          }; \addlegendentry{\DAC}
        \end{axis}
      \end{tikzpicture}
      \begin{tikzpicture}
        \begin{axis}[xlabel={Number of ext.\ constraints},
          % width=.24\linewidth,
        xmin=0, xmax=800, ymin=0, ymax=20000000,
          xtick={0, 200, 400, 600, 800},
          x=0.0037cm, y=0.45cm,
        legend pos=south east, 
        legend style={font=\scriptsize, draw=none, fill opacity=0.1, text opacity = 1},
        label style={font=\normalsize}, 
        tick label style={font=\normalsize},
        ymajorgrids=true,
        xmajorgrids=true,
        grid style=dashed,
        ymode = log
        ]
      
        \addplot[ mark size=1.5pt, color=blue, mark=square, ]
        coordinates {
         (50, 447454)(100,894909)(200,1789818)(400,3937600)(800,7875200)
        }; \addlegendentry{\DPPC}
        
        \addplot[ mark size=1.5pt, color=red, mark=triangle, ] coordinates
        {
          (50,34940)(100,55620)(200,95040)(400,164820)(800,279660)
        }; \addlegendentry{\DAC}
      \end{axis}
    \end{tikzpicture}
}
  \end{subcaptionbox}
  \begin{subcaptionbox}{\textsf{WS}\label{fig:ws}}[0.5\textwidth]{
     \begin{tikzpicture}
      \begin{axis}[xlabel={Number of agents}, %ylabel={CPU time (sec)},
        xmin=2, xmax=16, ymin=0, ymax=420000,
        xtick={2,4,8,12,16},
        x=0.232cm, y=0.53cm,
        legend pos=south east, legend style={font=\scriptsize, draw=none},
        label style={font=\normalsize}, tick label
        style={font=\normalsize},
         ymajorgrids=true,
        xmajorgrids=true,
        grid style=dashed, 
        ymode = log]

        \addplot[ mark size=1.5pt, color=blue, mark=square, ]
        coordinates {
         (2,5157)(4,30996)(8,97950)(12,269139)(16,418233)
        }; \addlegendentry{\DPPC}

         \addplot[ mark size=1.5pt, color=red, mark=triangle, ] coordinates
        {
         (2,2208)(4,7876)(8,18792)(12,25850)(16,46376)
       }; \addlegendentry{\DAC}
     \end{axis}
   \end{tikzpicture}
   \begin{tikzpicture}
      \begin{axis}[ xlabel={Number of tasks}, %ylabel={CPU time (sec)},
        xmin=80, xmax=480, ymin=0, ymax=1300000,
        xtick={80,160,240,320,400,480},
        x=0.008cm, y=0.53cm,
        legend pos=south east, legend style={font=\scriptsize, draw=none},
        label style={font=\normalsize}, tick label
        style={font=\normalsize}, 
        ymajorgrids=true,
        xmajorgrids=true,
        grid style=dashed,
        ymode = log
        ]

        \addplot[ mark size=1.5pt, color=blue, mark=square, ]
        coordinates {
        (80,25326)(160,89811)(240,174105)(320,417939)(400,761745)(480,1292391)
        }; \addlegendentry{\DPPC}

         \addplot[ mark size=1.5pt, color=red, mark=triangle, ] coordinates
        {
        (80,5430)(160,16284)(240,33924)(320,47670)(400,56760)(480,71470)
        }; \addlegendentry{\DAC}
      \end{axis}                               
    \end{tikzpicture}}
  \end{subcaptionbox}

 \caption{Evaluation of \DAC and \DPPC. The $y$-axis (on the
    log scale) represent the number of NCCCs.}\label{fig:exp}
\end{figure*}

\section{Evaluation}\label{sec:evaluation}

%----------------------------------------------

In this section we experimentally compare our algorithms against the
state-of-the-art algorithms for solving STNs. For centralized
algorithms, we compare our \AC algorithm against
\citeauthor{PlankenWK08}'s \PPC algorithm \shortcite{PlankenWK08}; for
distributed algorithms, we compare our \DAC algorithm against
\citeauthor{DBLP:journals/jair/BoerkoelD13}'s \DPPC algorithm
\shortcite{DBLP:journals/jair/BoerkoelD13}. All experiments for
distributed algorithms used an asynchronous simulator in which agents
are simulated by processes which communicate only through message
passing and default communication latency is assumed to be zero. Our
experiments were implemented in Python~3.6 and carried out on a
computer with an Intel Core i5 processor with a 2.9 GHz frequency per
CPU, 8 GB memory \footnote{The source code for our evaluation can be
  found in \url{https://github.com/sharingcodes/MaSTN}}.

As measures for comaring performances we use the number of constraint
checks and the number of non-concurrent constraint checks (NCCCs)
performed by the centralized algorithms and the distributed
algorithms, respectively. Given an STN $\network$, a constraint check
is performed when we compute relation
$r \gets I_{vw} \cap (I_{vu} \otimes I_{uw})$ and check if
$r = I_{vw}$ or $r \not\subseteq I_{vw}$.
%In the case of distributed algorithms, this measure is extended to the number of \emph{none concurrent constraint checks} which means the number of constraint checks that happen sequentially. 
%The second measure concerns the CPU time and is strongly correlated
%with the first one, as the runtime of algorithms relies heavily on the
%number of \emph{none concurrent constraint checks} (NCCCs) performed.

\subsection{\AC vs. \PPC}
\subsubsection{Datasets}
We selected instances from the benchmark datasets of STNs used in
\cite{planken2012computing} for evaluations. We considered the
scale-free graphs (\textsf{Scale-free-1}) with 1000 vertices and
density parameter varying from 2 to 50. We also considered the
scale-free graphs (\textsf{Scale-free-2}) with varying vertex count.
The scale-free density parameter for this selection is 5. Beside these
artificially constructed graphs, we also considered graphs that are
based on the road network of New York City (\textsf{New York}). This
dataset contains 170 graphs on 108--3906 vertices, 113--6422 edges.

\subsubsection{Results}
The results are presented in Figure~\ref{fig:exp_cen}, where base-10
log scales are used for the $y$-axes. For the scale-free graphs we
observe that \AC is 100--1000 times faster than \PPC. The dataset
\textsf{New York} only contains very sparse networks (each network's
density is less than $1\%$), thus both algorithms could easily solve
these networks. However, we still observe that \AC is about 5--12
times faster than \PPC.

\subsection{\DAC vs. \DPPC}

\subsubsection{Datasets}
We selected instances from the benchmark datasets of MaSTNs used in
\cite{DBLP:journals/jair/BoerkoelD13} for evaluations. The first
dataset \textsf{BDH} was randomly generated using the multiagent
adaptation of Hunsberger's \shortcite{hunsberger2002algorithms} random
STN generator. Each MaSTN has $N$ agents each with start time points
and end time points for 10 activities, which are subject to various
local constraints. In addition, each MaSTN has $X$ external
contraints. We evaluated the algorithms by varying the number of
agents $(N \in \{2, 4, 8, 12, 16\}, X = 50 \times (N-1))$ and the
total number of external constraints
$(N = 16, X \in \{100, 200, 400, 800\})$.

The second dataset \textsf{WS} is derived from a multiagent factory
scheduling domain \cite{wilcox2012optimization}, where $N$ agents are
working together to complete $T$ tasks in a manufacturing environment.
We evaluated algorithms by varying the number of agents
$(N \in \{2, 4, 8, 12, 16\}, T = 20 \times N)$ and the total number of
tasks $(N = 16, T \in \{80, 160, 240, 320, 400, 480\})$.

\subsubsection{Results}

The results are presented in Figure~\ref{fig:exp}, where base-10 log
scales are again used for the $y$-axes. For the \textsf{DBH} random
networks (Figure~\ref{fig:dbh}) we observe that \DAC is 5--30 times
faster than \DPPC. For the \textsf{WS} scheduling networks
(Figure~\ref{fig:ws}) \DAC is 2--10 times faster than \DPPC. For both
datasets we observe that, with increasing $x$-values, the $y$-values
(i.e., NCCCs) for \DAC grow slower than those for \DPPC.

\section{Conclusion}
In this paper we presented a novel AC-based approach for solving the
STP and the MaSTP. We have shown that arc-consistency is sufficient
for solving an STN. Considering that STNs are defined over infinite
domains, this result is rather surprising. Our empirical evaluations
showed that the AC-based algorithms are significantly more efficient
than their PC-based counterparts. This is mainly due to the fact
\magenta{that} PC-based algorithms add many redundant constraints in
the process of triangulation. More importantly, since our AC-based
approach does not impose new constraints between agents that are
previously not directly connected, it respects as much privacy of
these agents as possible. We should note here that even though our
distributed algorithm \DAC showed remarkable performance, it can be
further fine-tuned by using different termination detection
mechanisms~(cf.~\cite{mattern1987algorithms} and
\cite[Ch.~14]{raynal_distributed_2013}).

It would be interesting to see how the result in this paper can be
used for solving the general disjunctive temporal problems
\cite{STERGIOU200081}. Potential extensions of our paper also include
adapting our AC algorithms to incremental algorithms for the STP
\cite{planken_incrementally_2010}, dynamic situations
\cite{MorrisMV01} and uncertainty \cite{VenableY05}.

\section*{Acknowledgments}
\red{We thank the anonymous reviewers, who pointed out the connections to max/min-closed constraints and the shortest-path problem.
%helped improve this  paper. In particular, Remarks~\ref{rem:2} (iii) and \ref{rem:3} in  this paper are based on their comments. 
The work of SL was partially supported by NSFC (No. 11671244), and the work of JL was partially
  supported by the Alexander von Humboldt Foundation.}
\bibliographystyle{aaai}

\bibliography{reference}

\end{document}